\documentclass[reprint, aps, prl, superscriptaddress, nofootinbib
]{revtex4-1}
\usepackage{graphicx}
\usepackage{amsfonts}
\usepackage{amssymb}
\usepackage{amsthm}
\usepackage{array}
\usepackage{amsmath}
\usepackage{verbatim} 
\usepackage{hyperref}
\usepackage{color}
\usepackage{bbold}
\usepackage{epstopdf}
\usepackage{mathtools}
\usepackage{cancel}
\usepackage{soul}

\newtheorem{proposition}{Proposition}

\newtheorem{theorem}{Theorem}

\newtheorem{observation}{Observation}

\newcommand{\ket}[1]{\left\vert#1\right\rangle}
\newcommand{\bra}[1]{\left\langle#1\right\vert}

\def\bra#1{\langle #1|}
\def\ket#1{\left|#1 \right>}

\def\Tr{\mbox{Tr}}

\begin{document}
\title{Clock--Work Trade-Off Relation for Coherence in Quantum Thermodynamics}
\author{Hyukjoon Kwon}
\affiliation{Center for Macroscopic Quantum Control, Department of Physics and Astronomy, Seoul National University, Seoul, 151-742, Korea}
\author{Hyunseok Jeong}
\affiliation{Center for Macroscopic Quantum Control, Department of Physics and Astronomy, Seoul National University, Seoul, 151-742, Korea}
\author{David Jennings}
\affiliation{Department of Physics, University of Oxford, Oxford, OX1 3PU, United Kingdom}
\affiliation{QOLS, Blackett Laboratory, Imperial College London, London SW7 2AZ, United Kingdom}
\author{Benjamin Yadin}
\affiliation{QOLS, Blackett Laboratory, Imperial College London, London SW7 2AZ, United Kingdom}
\author{M. S. Kim}
\affiliation{QOLS, Blackett Laboratory, Imperial College London, London SW7 2AZ, United Kingdom}
\date{\today}
\begin{abstract}
In thermodynamics, quantum coherences---superpositions between energy eigenstates---behave in distinctly nonclassical ways. Here we describe how thermodynamic coherence splits into two kinds---``internal" coherence that admits an energetic value in terms of thermodynamic work, and ``external" coherence that does not have energetic value, but instead corresponds to the functioning of the system as a quantum clock. For the latter form of coherence we provide dynamical constraints that relate to quantum metrology and macroscopicity, while for the former, we show that quantum states exist that have finite internal coherence yet with zero deterministic work value. Finally, under minimal thermodynamic assumptions, we establish a clock-work trade-off relation between these two types of coherences. This can be viewed as a form of time-energy conjugate relation within quantum thermodynamics that bounds the total maximum of clock and work resources for a given system.
\end{abstract}
\maketitle

Classical thermodynamics describes the physical behavior of macroscopic systems composed of large numbers of particles. Thanks to its intimate relationship with statistics and information theory, the domain of thermodynamics has recently been extended to include small systems, and even quantum systems. One particularly pressing question is how the existence of quantum coherences, or superpositions of energy eigenstates, impacts the laws of thermodynamics \cite{Aberg14, Skrzypczyk14, Uzdin15, Korzekwa16}, in addition to quantum correlations \cite{Park13, Reeb14, Huber15, Llobet15}.

We now have a range of results for quantum thermodynamics \cite{Janzing2000, Horodecki13, Brandao13, Brandao15,  Lostaglio15, LostaglioX, Cwiklinski15, Goold16, Bera16, Wilming16, LostaglioPRL, Muller17, Gour17} that have been developed within the resource-theoretic approach. A key advantage of the resource theory framework is that it avoids highly problematic concepts such as ``heat" or ``entropy" as its starting point. Its results have been shown to be consistent with traditional thermodynamics \cite{Weilmann16} while allowing for the inclusion of coherence as a resource \cite{Sterltsov17} in the quantum regime. Very recently, a framework for quantum thermodynamics with coherence was introduced in Ref.~\cite{Gour17}. Remarkably, the thermodynamic structure (namely, which states $\hat\sigma$ are thermodynamically accessible from a given state $\hat\rho$) turns out to be fully describable in terms of a single family of entropies. This framework of ``thermal processes'' is defined by the following three minimal physical assumptions: (i) that energy is conserved microscopically, (ii) that an equilibrium state exists, and (iii) that quantum coherence has a thermodynamic value. These are described in more detail in \cite{Supplemental}. Note that thermal processes contain thermal operations (TOs) \cite{Janzing2000,Horodecki13} as a subset, and coincide with TOs on incoherent states; however, in contrast to TOs they admit a straightforward description for the evolution of states with coherences between energy eigenspaces.

In this Letter we work under the same thermodynamic assumptions (i)--(iii) as above and show that quantum coherence in thermodynamics splits into two distinct types: \emph{internal} coherences between quantum states of the same energy, and \emph{external} coherences between states of different energies. This terminology is used because the external coherences in a system are only defined relative to an external phase reference frame, while internal coherences are defined within the system as relational coherences between its subcomponents.

We focus on the case of an $N$-partite system with noninteracting subsystems.
The Hamiltonian is written as $\hat{H} = \sum_{i=1}^N \hat{H}_i$ and we assume that each $i$th local Hamiltonian $\hat{H}_i$ has an energy spectrum $\{ E_i \}$ with local energy eigenstates $\ket{E_i}$.
Then a quantum state of this system may be represented as 
$$
\hat\rho =   \sum_{{\boldsymbol E}, {\boldsymbol E'}  }
 \rho_{{\boldsymbol E}{\boldsymbol E'}  } \ket{\boldsymbol E} \bra{\boldsymbol E'},
$$
where ${\boldsymbol E} := (E_1, E_2, \cdots, E_N)$ and $\ket{\boldsymbol E} := \ket{E_1 E_2 \cdots E_N}$.
We also define the total energy of the string  ${\boldsymbol E}$ as ${\cal E}_{\boldsymbol E} := \sum_{i=1}^N E_i$.
Classical thermodynamic properties are determined by the probability distribution of the local energies, including their correlations.
This information is contained in the diagonal terms of density matrix
$P(\boldsymbol E)  := {\rm Tr} \left[ \hat\Pi_{\boldsymbol E} \hat\rho \right] = \rho_{\boldsymbol E \boldsymbol E}$,
where $\hat\Pi_{\boldsymbol E} := \ket{\boldsymbol E} \bra{\boldsymbol E}$.
Corresponding classical states could have degeneracies in energy, but we still have a distinguished orthonormal basis set $\{ \ket{\mathbf{E}} \}$.
The probability distribution of the total energy is $\displaystyle p_{\cal E} := \sum_{\boldsymbol{E} \colon {\cal E}_{\boldsymbol E} = {\cal E}} P(\boldsymbol E)$. So every state has a corresponding classical state defined via the projection $\Pi(\hat\rho) := \sum_{\boldsymbol E} P(\boldsymbol E) \ket{\boldsymbol E} \bra{\boldsymbol E}$.

However, a quantum system is defined by more than its classical energy distribution -- it may have coherence in the energy eigenbasis. This coherence is associated with nonzero off-diagonal elements in the density matrix, namely
$\ket{\boldsymbol E} \bra{\boldsymbol E'}$ for $\boldsymbol E  \neq \boldsymbol E'$.
The internal coherence corresponds to off-diagonal terms of the same total energy (where $\cal{E}_{\boldsymbol E} = \cal{E}_{\boldsymbol E'}$) and external coherence corresponds to terms with different energies ($\cal{E}_{\boldsymbol E} \neq \cal{E}_{\boldsymbol E'}$). For any state $\hat\rho$, we denote the corresponding state in which all external coherence is removed by $\mathcal{D}(\hat\rho) := \sum_{\mathcal E} \hat\Pi_{\mathcal E} \hat\rho \hat\Pi_{\mathcal E}$, where $\hat\Pi_{\mathcal E} := \sum_{\boldsymbol{E} \colon \mathcal{E}_{\boldsymbol{E}} = \mathcal{E}} \hat\Pi_{\boldsymbol E}$ is the projector onto the eigenspace of total energy $\mathcal{E}$.

As illustrated in Fig.~\ref{FIG1}, internal coherence may be used to extract work; however, it has been shown that external coherences obey a superselection rule (called ``work locking") that forbids work extraction, and is unavoidable if one wishes to explicitly account for all sources of coherence in thermodynamics \cite{Lostaglio15}. We study this phenomenon by defining the process of extracting work \emph{purely from the coherence}, without affecting the classical energy statistics. We find the conditions under which work can be deterministically extracted in this way from a pure state. Next, we show that external coherence is responsible for a system's ability to act as a clock. The precision of the clock may be quantified by the quantum Fisher information (QFI) \cite{Braunstein94}; we show that the QFI satisfies a second-law-like condition, stating that it cannot increase under a thermal process. Finally, we derive a fundamental trade-off inequality between the QFI and the extractable work from coherence demonstrating how a system's potential for producing work is limited by its ability to act as a clock and vice versa.

{\it Extractable work from coherence.---}
Here we demonstrate that, in a single-shot setting, work may be extracted from coherence without changing the classical energy distribution $P(\boldsymbol E)$ of the system.
We consider the following type of work extraction process
$$
\hat\rho \otimes \ket{0} \bra{0}_B \xrightarrow{\rm thermal~process} \Pi(\hat\rho) \otimes \ket{W}\bra{W}_B,$$ in which the energy of a work qubit \cite{footnote} $B$ with Hamiltonian $\hat{H}_B = W \ket{W}\bra{W}_B$ ($W \geq 0$) is raised from $\ket{0}_B$ to $\ket{W}_B$.

For energy-block-diagonal states $\hat\rho = \mathcal{D}(\hat\rho)$ and $\hat\sigma = \mathcal{D}(\hat\sigma)$, the
work distance is defined as
$D_{\rm work} (\hat\rho \succ \hat\sigma)
= \displaystyle \inf_{\alpha} \left[ F_\alpha(\hat\rho) - F_\alpha(\hat\sigma) \right] $ \cite{Brandao15},
where $F_\alpha(\hat\rho) = k_B T S_\alpha(\hat\rho || \hat\gamma ) - k_B T \log Z$ is a generalized free energy based on the R\'enyi divergence
$$
S_\alpha({\hat{\rho}}||{\hat{\sigma}}) = 
\begin{cases}
\frac{1}{\alpha-1} \log \Tr [ {\hat{\rho}}^\alpha {\hat{\sigma}}^{1-\alpha} ], & \alpha \in [0,1)\\
\frac{1}{\alpha-1} \log \Tr \left[ \left({\hat{\sigma}}^\frac{1-\alpha}{2\alpha} {\hat{\rho}} {\hat{\sigma}}^\frac{1-\alpha}{2\alpha}\right)^\alpha \right], & \alpha>1.
\end{cases}
$$
Here $Z = {\rm Tr} e^{-\hat{H}/(k_B T)}$ is a partition function of the system.
The work distance is the maximum extractable work by a thermal process by taking $\hat\rho$ to $\hat\sigma$ \cite{Brandao15}.
\begin{figure}[t]
\includegraphics[width=0.8\linewidth]{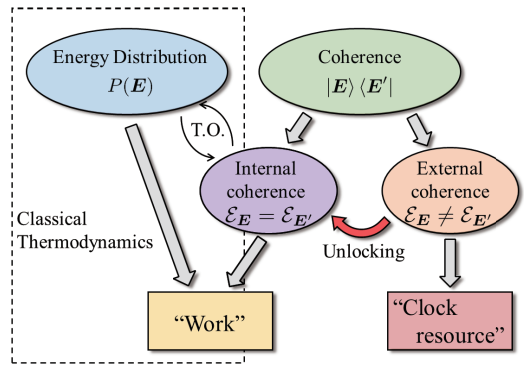}
\caption{Thermodynamic resources for many-body quantum systems. Coherences between energy levels provide coherent oscillations and are resources for the composite system to act as a quantum clock. At the other extreme, projective energy measurements on the individual systems provide the classical energy statistics, which may display classical correlations.  Intermediate between these two cases are quantum coherences that are internal to energy eigenspaces.  Partial interconversions are possible between these three aspects.}
\label{FIG1}
\end{figure}

\begin{figure}[t]
\includegraphics[width=\linewidth]{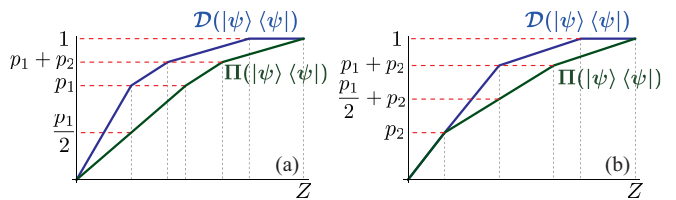}
\caption{Thermomajorization graph for the energy-block-diagonal state ${\cal D}(\ket{\psi}\bra{\psi})$ and its projection to an incoherent state ${\Pi}(\ket{\psi}\bra{\psi})$ for a two-qubit state $\ket{\psi}$ studied in the text with coefficients $p_1$ and $p_2$. $Z$ represents the partition function of the system.
(a) When $p_1 e^{\beta \omega_0}$ is the maximum among $p_i e^{\beta E_i}$, $W_{\rm coh}$ is positive, but (b) if another energy (e.g., $p_2 e^{2\beta\omega_0}$ in the plot) obtains the maximum, $W_{\rm coh}=0$.
}
\label{WcohQ}
\end{figure}

Even when the initial state $\hat\rho$ is not block diagonal in the energy basis, the extractable work is still given by $D_{\rm work}(\mathcal{D}(\hat\rho) \succ \hat\sigma)$, so external coherence cannot be used to extract additional work \cite{Lostaglio15}. In order to exploit external coherence for work, one needs multiple copies of $\hat\rho$ \cite{Brandao13, Korzekwa16} or ancilliary coherent resources \cite{Brandao13, Aberg14, Lostaglio15}. Thus, the single-shot extractable work purely from coherence is given by
\begin{equation}
\label{Wcoh}
W_{\rm coh} = \inf_\alpha \left[ F_\alpha({\cal D} (\hat\rho)) - F_\alpha(\Pi(\hat\rho)) \right].
\end{equation}

For example, consider extracting work from coherence in the pure two-qubit state
$$
\ket{\psi} = \sqrt{p_0} \ket{00} +  \sqrt{p_1} \left( \frac{ \ket{01} + \ket{10}}{ \sqrt{2} } \right)+ \sqrt{p_2} \ket{11},
$$
where each qubit has local Hamiltonain $H_i = \omega_0 \ket{1}\bra{1}$. As shown in Fig.~\ref{WcohQ} using the concept of thermomajorization \cite{Horodecki13}, we have $W_{\rm coh} >0$ only for sufficiently large $p_1$. In this case, we have the necessary condition $p_1 > (1+e^{\beta\omega_0} + e^{-\beta\omega_0})^{-1}$
and the sufficient condition $p_1 > e^{\beta\omega_0}/({1+e^{\beta\omega_0}})$ for $W_{\rm coh} > 0$, independent of $p_0$ and $p_2$ at an inverse temperature $\beta = (k_B T)^{-1}$. See  Ref.~\cite{Supplemental} for details.

We generalize this statement to the internal coherence of an arbitrary pure state. Since external coherences cannot contribute to work, we need only consider the properties of the state dephased in the energy eigenbasis.
\begin{observation}
\label{PureWext}
Consider any pure state state $\ket{\psi}$ with $\mathcal{D}( \ket{\psi}\bra{\psi}) =  \sum_k p_\mathcal{E}|\psi_\mathcal{E}\rangle\langle \psi_\mathcal{E}|$, where $|\psi_\mathcal{E}\rangle$ is an energy $\mathcal{E}$ eigenstate. Nonzero work can be extracted deterministically from the internal coherence of $\hat\rho$ at an inverse temperature $\beta$ if and only if $\Pi\left(\ket{\psi_{{\cal E}^*}}\bra{\psi_{{\cal E}^*}}\right) \neq \ket{\psi_{{\cal E}^*}}\bra{\psi_{{\cal E}^*}}$ for ${\cal E}^* = \underset{\cal E}{\rm argmax}~p_{\cal E}e^{\beta {\cal E}}$.
\end{observation}
\noindent Internal coherence has some overlap with nonclassical correlations, namely quantum discord \cite{Modi12}. Consider the following quantity which quantifies the sharing of free energy between subsystems:
$$
 C_\alpha(\hat\rho_{1:2:\cdots:N}) := \beta \left[ F_\alpha(\hat\rho) - \sum_{i=1}^N F_\alpha(\hat\rho_i) \right],
$$
where $\hat \rho_i$ is the local state of the $i$th subsystem. For nondegenerate local Hamiltonians, the extractable work from coherence can be written as
$$
W_{\rm coh} = k_B T \inf_{\alpha} \left[ C_\alpha( {\cal D}(\hat\rho)) - C_\alpha( \Pi(\hat\rho)) \right],
$$
noting that the local free energies are the same for ${\cal D}(\hat\rho)$ and $\Pi(\hat\rho)$. This is of the same form as discord defined by Ollivier and Zurek \cite{Zurek01}, expressed as a difference between total and classical correlations. Note that the classical correlations are defined here with respect to the energy basis, instead of the usual maximization over all local basis choices. This free-energy correlation is also related to  ``measurement-induced disturbance'' \cite{Luo08} by considering $\Pi(\hat\rho)$ as a classical measurement with respect to the local energy bases.

Unlike previous related studies \cite{Llobet15,Korzekwa16}, our result requires that only coherence is consumed in the work extraction processes, leaving all energy statistics unchanged. We may also consider the ``incoherent" contribution to the extractable work, $W_{\rm incoh} := \inf_\alpha [F_\alpha(\Pi(\hat\rho)) - F_\alpha(\hat\gamma)] = F_{0}(\Pi(\hat\rho)) + k_B T \log Z$, which is the achievable work from an incoherent state $\Pi(\hat\rho)$ ending with a Gibbs state $\hat\gamma$. The sum of the coherent and incoherent terms cannot exceed the total extractable work from $\hat\rho$ to $\hat\gamma$, i.e., $W_{\rm coh} + W_{\rm incoh} \leq W_{\rm tot} = D_{\rm work}({\cal D}(\hat\rho) \succ \hat\gamma) $. The equality holds when $W_{\rm coh}$ in Eq.~(\ref{Wcoh}) is given at $\alpha = 0$.
We also point out that this type of work extraction process operates without any measurement or information storage as in Maxwell's demon \cite{Lloyd97, Zurek03} or the Szilard engine \cite{Park13} in the quantum regime.

Apart from the above example, a significant case is the so-called coherent Gibbs state \cite{Lostaglio15}, defined for a single subsystem as $\ket{\gamma} := \sum_i \sqrt{\frac{e^{-\beta E_i}}{Z}} \ket{E_i}$. No work can be extracted from this state, as $\mathcal{D}(\ket{\gamma}\bra{\gamma}) = \hat\gamma$ -- an instance of work locking. However, nonzero work can be unlocked \cite{Lostaglio15} from multiple copies $\ket{\gamma}^{\otimes N},\, N>1$. In fact, from Observation \ref{PureWext}, we see that $N=2$ is always sufficient to give $W_{\rm coh}>0$. This is because $p_{\mathcal{E}} e^{\beta \mathcal{E}}$ is proportional to the degeneracy of the $\mathcal{E}$ subspace -- there always exists a degenerate subspace for $N \geq 2$, and this is guaranteed to have coherence.

{\it Coherence as a clock resource.---}
Having discussed the thermodynamical relevance of internal coherence, we now turn to external coherence. Suppose we have an initial state $\hat\rho_0 = \sum_{{\boldsymbol E}, {\boldsymbol E'}} \rho_{{\boldsymbol E}{\boldsymbol E'}  } \ket{\boldsymbol E} \bra{\boldsymbol E'}$. After free unitary evolution for time $t$, this becomes  $\hat\rho_t =   \sum_{{\boldsymbol E}, {\boldsymbol E'}  } \rho_{{\boldsymbol E}{\boldsymbol E'}  } e^{- i \Delta \omega_{\boldsymbol E \boldsymbol E'}  t} \ket{\boldsymbol E} \bra{\boldsymbol E'}$, where each off-diagonal component $\ket{\boldsymbol E} \bra{\boldsymbol E'}$ rotates at frequency $\Delta \omega_{\boldsymbol E \boldsymbol E'} = ({\cal E}_{\boldsymbol E} - {\cal E}_{\boldsymbol E'})/\hbar$. Internal coherences do not evolve ($\Delta \omega_{\boldsymbol E \boldsymbol E'} = 0$), while external coherences with larger energy gaps, and hence higher frequencies, can be considered as providing more sensitive quantum clocks \cite{Chen10,Komar14}. 

By comparing $\hat\rho_0$ with $\hat\rho_t$, one can estimate the elapsed time $t$. More precisely, the resolution of a quantum clock can be quantified by $(\Delta t)^2 = \langle (\hat{t} - t)^2 \rangle$, where $\hat{t}$ is the time estimator derived from some measurement on $\hat\rho_t$.
The resolution is limited by the quantum Cram\'er-Rao bound \cite{Braunstein94},
$(\Delta t)^2 \geq 1 / I_F(\hat{\rho}, {\hat{H})}$,
where $I_F(\hat{\rho}, \hat{H}) = 2 \sum_{i,j} \frac{(\lambda_i - \lambda_j)^2}{\lambda_i + \lambda_j} |\bra{i}  \hat{H} \ket{j}|^2$ is the QFI, and $\lambda_i,\,  \ket{i}$ are the eigenvalues and eigenstates of $\hat\rho$, respectively.
For the optimal time estimator $\hat{t}$ saturating the bound, the larger the QFI, the higher the clock resolution.
The maximum value of QFI for a given Hamiltonian $\hat{H}$ is $(E_{\rm max} - E_{\rm min})^2$, which can be obtained by the equal superposition $\ket{E_{\rm min}} + \ket{E_{\rm max}}$ between the maximum ($E_{\rm max}$) and minimum ($E_{\rm min}$) energy eigenstates. The Greenberger-Horne-Zeilinger (GHZ) state in an $N$-particle two-level system is a state of this form.

Another family of relevant measures of the clock resolution is the skew information $I_\alpha({\hat{\rho}},  \hat{H}) = \Tr({\hat{\rho}}  \hat{H}) - \Tr({\hat{\rho}}^\alpha  \hat{H} {\hat{\rho}}^{1-\alpha}  \hat{H})$ 
for $0 \leq \alpha \leq 1$ \cite{Wigner63,Braunstein94}. For pure states, both the QFI and skew information reduce to the variance: $\frac{1}{4} I_F(\ket{\psi},\hat{H}) = I_\alpha(\ket{\psi},\hat{H}) = V(\ket{\psi},\hat{H}) := \bra{\psi}{\hat{H}^2}\ket{\psi} - \bra{\psi}\hat{H}\ket{\psi}^2$.
In particular, the skew information of $\alpha=1/2$ has been studied in the context of quantifying coherence \cite{Luo17} and quantum macroscopicity \cite{Kwon18}. 
We also remark that a similar approach to ``time references" in quantum thermodynamics has been recently suggested using an entropic clock performance quantifier \cite{Gour17}.

We first note that, even though a quantum state might be very poor at providing work, it can still function as a good time reference.
The coherent Gibbs state is a canonical example. As mentioned earlier, no work may be extracted from $\ket{\gamma}$; however, such a state does allow time measurements, since $I_F(\ket{\gamma}, \hat{H}) =  4 \frac{\partial^2}{\partial \beta^2} \log Z$, which is proportional to the heat capacity $k_B \beta^2 \frac{\partial^2}{\partial \beta^2} \log Z$ \cite{Crooks12}.

Furthermore, the QFI and skew information are
based on monotone metrics \cite{Hansen08, Petz11},
and monotonically decrease under time-translation-covariant operations \cite{Marvian14,Yadin16}. 
It follows that the resolution of a quantum clock gives an additional constraint of a second-law type.

\begin{observation}
Under a thermal process, the quantum Fisher (skew) information $I_{F(\alpha)}$ of a quantum system cannot increase, i.e.
\begin{equation}
\label{Imonotone}
\Delta I_{F(\alpha)} \leq 0.
\end{equation}
\end{observation}
We highlight that this condition is independent from those obtained previously, based on a family of entropy asymmetry measures
$ A_\alpha({\hat{\rho}}_S) = S_\alpha ({\hat{\rho}}_S || {\cal D} ({\hat{\rho}}_S))$ \cite{Lostaglio15}
and modes of asymmetry $\sum_{{\cal E}_{\boldsymbol E} - {\cal E}_{\boldsymbol E'} = \omega} |\rho_{\boldsymbol E \boldsymbol E'}| $  \cite{LostaglioX, MarvianPRA}. In Ref.~\cite{Supplemental}, we give an example of a state transformation that is forbidden by Eq.~(\ref{Imonotone}) but not by previous constraints.

Importantly, condition Eq.~(\ref{Imonotone}) remains significant in the many-copy or independent and identically distributed limit. This follows from the additivity of the QFI and skew information, namely, $ I_{F(\alpha)} ({\hat{\rho}}^{\otimes N}, \hat{H})/ N = I_{F(\alpha)}({\hat{\rho}},\hat{H}_1)$ for all $N$, where $\hat{H}=\sum_{i=1}^N \hat{H}_i$. In contrast, the measure $A_\alpha$ is negligible in this limit: $\displaystyle \lim_{N\rightarrow \infty} A_\alpha({\hat{\rho}}^{\otimes N}) / N = 0$ for all $\alpha$  \cite{Lostaglio15}.
The free energy $F_\alpha (\hat\rho)$ with $\alpha=1$ has been stated to be the unique monotone for asymptotic transformations \cite{Brandao13}. However, this is true \emph{only} if one is allowed to use a catalyst containing external coherence between every energy level, $\ket{M} = |M|^{-1/2} \sum_{m \in M} \ket{m}$, where $M= \{ 0,  \dots , 2N^{2/3} \}$, which contains a superlinear amount of clock resources $I_F= O(N^{4/3})$, so $I_F/N$ is unbounded. Thus, Eq.~(\ref{Imonotone}) is the first known nontrivial coherence constraint on asymptotic transformations under thermal processes without additional catalytic coherence resources.

We can illustrate the physical implications of this condition in an $N$-particle two-level system with a local Hamiltonian $\hat{H}_i = 0 \ket{0}_i\bra{0} + \omega_0 \ket{1}_i\bra{1}$ for every $i$th particle.
As noted above, for a product state ${\hat{\rho}}^{\otimes N}$, the QFI and skew information scale linearly with $N$.
On the other hand, the GHZ state $\ket{\psi_{\rm GHZ}} = 2^{-1/2}(\ket{0}^{\otimes N} + \ket{1}^{\otimes N})$ has quadratic scaling, $I_{F(\alpha)} (\ket{\psi_{\rm GHZ}}, \hat{H}) = O(N^2)$.
Thus the restriction given by Eq.~(\ref{Imonotone}) indicates that a thermal process cannot transform a product state into a GHZ state.
More generally, it is known that $I_F(\hat\rho, \hat{H}) \leq kN$ for $k$-producible states  in $N$-qubit systems \cite{Toth12,Hyllus12}, so
genuine multipartite entanglement is necessary to achieve a high clock precision of $I_F = O(N^2)$. Also note that the QFI has been used to quantify ``macroscopicity", the degree to which a state displays quantum behavior on a large scale \cite{Frowis2012,Yadin16}.

{\it Trade-off between work and clock resources.---} 
Having examined the two types of thermodynamic coherence independently, it is natural to ask if there is a relation between them.
Here, we demonstrate that there is always a trade-off between work and clock coherence resources.
We first give the following bound in an $N$-particle two-level system:
\begin{theorem} [Clock/work trade-off for two-level subsystems] 
For a system composed of N two-level particles with energy level difference $\omega_0$, the coherent work and clock resources satisfy 
\begin{equation}
\label{TradeOff}
W_{\rm coh} \leq N k_B T (\log 2) H_b\left( \frac{1}{2} \left[1-\sqrt{\frac{I_F(\hat\rho, \hat{H})}{N^2 \omega_0^2}} \right] \right) ,
\end{equation}
where $\hat{H}$ is the total Hamiltonian and $H_b(r) = -r \log_2 r -(1-r)\log_2(1-r) $ is the binary entropy. 
\label{TradeOffTheorem}
\end{theorem}
This shows that a quantum state cannot simultaneously contain maximal work and clock resources. When the clock resource is maximal, $I_F = N^2 \omega_0^2$, no work can be extracted from coherence $W_{\rm coh} = 0$. Conversely, if the extractable work form coherence is maximal, $W_{\rm coh} = N k_B T \log 2$, the state cannot be utilized as a quantum clock as $I_F=0$.
For $N=2$ we derive a tighter inequality:
\begin{equation}
\label{QubitTradeOff}
W_{\rm coh} + ( k_B T  \log 2) \left( \frac{ I_F(\hat\rho, \hat{H})}{4\omega_0^2} \right) \leq k_B T  \log 2 .
\end{equation}

\begin{figure}[t]
\includegraphics[width=0.9\linewidth]{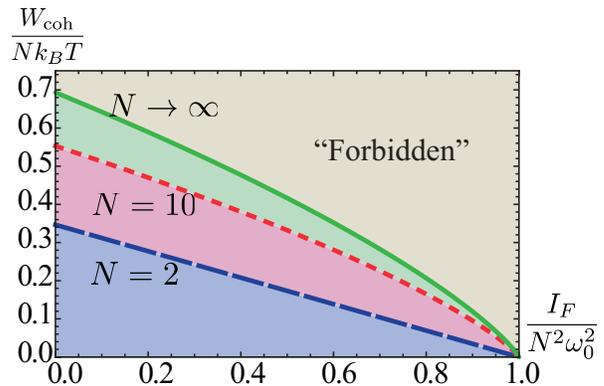}
\caption{Trade-off between work and clock coherences. The solid line refers to Eq.~(\ref{TradeOff}), the dashed line refers to Eq.~(\ref{QubitTradeOff}), and the dotted line refers the tighter bound for $N=10$. }
\label{TradeOffFig}
\end{figure}

We demonstrate that the GHZ state $\ket{\psi_{\rm GHZ}} = 2^{-1/2} (\ket{0}^{\otimes N} + \ket{1}^{\otimes N})$ and Dicke states $\ket{N,k} = \binom{N}{k}^{-1/2}\sum_{P} P ( \ket{1}^{k} \ket{0}^{N-k})$, summing over all permutations $P$ of subsystems, are limiting cases that saturate this trade-off relation.
For a Dicke state, the extractable work from coherence is given by $W_{\rm coh} = k_B T \log \binom{N}{k} \approx  N k_B T (\log 2) H_b(k/N)$.
However, each Dicke state has $I_F = 0$ since it has support on a single energy eigenspace with ${\cal E} = k\omega_0$.
In particular, when $k=N/2$, $W_{\rm coh} = N k_B T \log 2$, attaining the maximal value and saturating the  bounds Eqs.~(\ref{TradeOff}) and (\ref{QubitTradeOff}).
The GHZ state behaves in the opposite way: $\ket{\psi_{\rm GHZ}}$ has maximal QFI $I_F(\ket{\psi_{\rm GHZ}}, \hat{H}) = N^2 \omega_0^2$
while having no internal coherence; thus, $W_{\rm coh}=0$.
In this case, we can see the saturation of both bounds Eqs. (\ref{TradeOff}) and (\ref{QubitTradeOff}).

Furthermore, our two-level trade-off relation can be generalized for an arbitrary noninteracting $N$-particle system.
\begin{theorem}[Clock/work trade-off for arbitrary subsystems]
\label{TradeOffN}
Let $\hat H$ be a noninteracting Hamiltonian of $N$ subsystems, where the $i$th subsystem has an arbitrary (possibly degenerate) $d^{(i)}$-level spectrum $\{ E_1^{(i)} \leq E_2^{(i)} \leq \cdots \leq E_{d^{(i)}}^{(i)} \}$. Also define $\Delta_E^2 = \sum_{i=1}^N (\Delta_E^{(i)})^2 $ with $\Delta_E^{(i)} = E_{d^{(i)}}^{(i)} - E_1^{(i)}$. Then
\begin{equation} \label{GeneralTradeOff}
W_{\rm coh} + k_B T \left( \frac{I_F(\hat\rho,\hat{H})}{2 \Delta_E^2} \right) \leq k_B T \sum_{n=1}^N \log d^{(n)}.
\end{equation}
\end{theorem}

This is more generally applicable than Eq.~(\ref{QubitTradeOff}), but is weaker for two-level subsystems -- maximal $I_F$ does not imply $W_{\rm coh}=0$ via Eq.~(\ref{GeneralTradeOff}).
  
For systems with identical local $d$-level Hamiltonians, Eq.~(\ref{GeneralTradeOff}) reduces to
\begin{equation}
\bar{w}_{\rm coh} + k_B T \left( \frac{I_F(\hat\rho,\hat{H})}{2N^2 \Delta_0^2} \right) \leq k_B T \log d ,
\end{equation}
where $\bar{w}_{\rm coh} = W_{\rm coh} / N$ is extractable work per particle and $\Delta_E^2 = N \Delta_0^2$ where $\Delta_0$ is the maximum energy difference between the local energy eigenvalues. Our bounds do not limit the extractable work in the independent and identically distributed limit, since $I_F(\hat \rho^{\otimes N},\hat H) / N^2 \to 0$ for $N \gg 1$.

We can also describe how one extends our analysis into the regime of weak interactions between local systems. We note that interactions can break degeneracies in energy eigenspaces, and energy eigenstates of the free Hamiltonian may not be expressed as a product of local states. However breaking of degeneracies in energy should only be treated above a finite experimental width $\epsilon$ in energy resolution. Weak perturbations admit a similar analysis in terms of work extraction up to $\epsilon$ fluctuations. This $\epsilon$-energy window allows us to use external coherences with energy gaps less than $\epsilon$ for work extraction by effectively treating them as internal coherences in the same energy levels.
In this case, we can calculate how work and clock resources are perturbed and note that the trade-off relation Eq.~(\ref{GeneralTradeOff}) still holds with an ${\cal O}(\epsilon)$ correction.
We also discuss an example for transverse Ising models in which the trade-off relation can be resolved using a quasiparticle picture (see Ref.~\cite{Supplemental} for details).
The more general interacting case is nontrivial and we leave this for future study.

%In the case of weak interaction, we may allow $\epsilon$ amount of gap to be in essentially the same energy to use internal coherences for approximate work extraction \cite{Brandao15}. We note that the trade-off relation (\ref{GeneralTradeOff}) still holds with the correction of $O(\epsilon)$ (see \cite{Supplemental} for details and for discussions on how the problem of nonlocal energy-diagonal states can be circumvented in a simple example of the Ising model).

{\it Remarks.---}
We have found that thermodynamic coherence in a many-body system can be decomposed into time- and energy-related components.
Many-body coherence contributing to the thermodynamic free energy has been shown to be convertible into work by a thermal process,
without changing the classical energy statistics.
We have illustrated that this work-yielding resource comes from correlations due to coherence in a multipartite system.
We have also shown that coherence may take the form of a clock resource, and we have quantified this with the quantum Fisher (skew) information.
Our main result is a trade-off relation between these two different thermodynamic coherence resources.

%{\it Acknowledgements.---}
\begin{acknowledgments}
{\it Acknowledgements.---}
This work was supported by the UK EPSRC (EP/KO34480/1) the Leverhulme Foundation (RPG-2014-055), 
the NRF of Korea grant funded by the Korea government (MSIP) (No. 2010-0018295),
and the Korea Institute of Science and Technology Institutional Program (Project No. 2E26680-16-P025).
D. J. and M.~S.~K. were supported by the Royal Society.
\end{acknowledgments}

\appendix

\widetext
\newpage
\section{Supplemental Material}

\section{Physical assumptions for the analysis.}

For the class of free operations that define the thermodynamic framework, we make the following three physical assumptions \cite{Gour17}.
\begin{enumerate}
\item \textit{Energy is conserved microscopically.} We assume that any quantum operation $\mathcal{E}: \mathcal{B}(\mathcal{H}_A) \rightarrow \mathcal{B}(\mathcal{H}_{A'})$ that is thermodynamically free admits a Stinespring dilation of the form
\begin{equation}\label{micro}
\mathcal{E}(\hat{\rho}_A) = {\rm tr}_C \hat{V} (\hat{\rho}_A \otimes \hat{\sigma}_B) \hat{V}^\dagger
\end{equation}
where the isometry $\hat{V}$ conserves energy microscopically, namely we have
\begin{equation}
\hat{V} (\hat{H}_A \otimes \mathbb{I}_B + \mathbb{I}_A \otimes \hat{H}_B) =  (\hat{H}_{A'} \otimes \mathbb{I}_C + \mathbb{I}_{A'} \otimes \hat{H}_C) \hat{V}.
\end{equation}
 Here $H_S$ denotes the Hamiltonian for the system $S \in \{ A,A', B, C\}$. 
\item \textit{An equilibrium state exists.} We assume that for any systems $A$ and $A'$ there exist states $\hat{\gamma}_A$ and $\hat{\gamma}_{A'}$ such that $\mathcal{E} ( \hat{\gamma}_A) = \hat{\gamma}_{A'}$ for all thermodynamically free operations $\mathcal{E}$ between these two quantum systems. In the case where one admits an unbounded number of free states within the theory, this together with energy conservation essentially forces one to take 
\begin{equation}
\hat{\gamma}_A = \frac{1}{Z_A} e^{-\beta \hat{H}_A},
\end{equation}
at some temperature $T = (k_B\beta)^{-1}$ and $Z_A = {\rm tr} \left[ e^{-\beta \hat{H}_A} \right]$. However one may also consider scenarios in which the thermodynamic equilibrium state deviates from being a Gibbs state. Here we restrict our analysis to thermal Gibbs states.
\item \textit{Quantum coherences are not thermodynamically free.} Since one is interested in quantifying the effects of coherence in thermodynamics, one must not view coherence as a free resource that can be injected into a system without being included in the accounting. We therefore assume that the thermodynamically free operations do not smuggle in coherences in the following sense: if $\mathcal{E}$ has a microscopic description of the form \ref{micro} then the same operation is possible with $\hat{\sigma}_B \rightarrow \mathcal{D} (\hat{\sigma}_B)$ and $\hat{V} \rightarrow \hat{W}$ being some other energy conserving isometry. In other words \emph{no coherences in $\hat{\sigma}_B$ are exploited for free}. It can be shown this assumption has the mathematical consequence that
\begin{equation}
\mathcal{E}( e^{-it \hat{H}_A} \hat{\rho}_A e^{it \hat{H}_A}) = e^{-i t \hat{H}_{A'}} \mathcal{E}(\hat{\rho}_A) e^{i t \hat{H}_{A'}},
\end{equation}
for any translation through a time interval $ t \in \mathbb{R}$, namely covariance under time-translations.
\end{enumerate}
The first assumption is simply a statement of energy conservation, while the second assumption singles out a special state that is left invariant under the class of thermodynamic processes. Note that this equilibrium state could be a \emph{non-Gibbsian} state and the formalism in \cite{Gour17} would still apply, however for our analysis here, we shall take the state to be a thermal Gibbs state, simply because this provides us with a notion of temperature and a comparision with traditional equilibrium thermodynamics. Finally the last assumption can be understood as a criterion for non-classicality for the framework, and the set of quantum operations defined by these three physical assumptions is called (generalized) \emph{thermal processes} (TPs). By using these three physical assumptions (1--3), one can show that a set of TPs coincide with the set of time-translational covariant Gibbs-preserving maps \cite{Gour17}.

We emphasize that TPs are different from thermal operations (TOs)  \cite{Janzing2000,Horodecki13}. A key difference is that TOs are defined by a energy conserved unitary operation between a system and a equilibrium bath $\hat\gamma_B = Z_B^{-1} e^{-\beta \hat{H}_B}$.
$${\rm tr}_B \hat{V} (\hat\rho_A \otimes \hat\gamma_B) \hat{V}^\dagger,$$
satisfying $[\hat{V}, \hat{H}_A + \hat{H}_B] =0$. Crucially, TOs fix the auxiliary system to be in the Gibbs state. One might think that the assumptions (1--3) imply that TP coincides with TO but this is not the case. The operations in TO all obey (1--3) and so $TO \subseteq TP$, however there exist transformations in TP that are not in TO. To show this, consider the regime in which all Hamiltonians are trivial, namely $\hat{H}=0$. For this scenario the equilibrium Gibbs state becomes the maximally mixed state $\frac{1}{d} \mathbb{1}$ and conditions (1) and (3) are trivially true for any map. The set of TOs for this situation coincide with the set of \emph{noisy operations}, while the set of TPs coincide with the set of \emph{unital maps}. It is known that these two sets are not the same.

However it is known that the \emph{state interconversion structure} of noisy operations \emph{coincides} with that of unital maps: we have $\hat\rho \longrightarrow \hat\sigma$ under a noisy operation if and only if it is possible under a unital operation. Therefore one might conjecture that these two classes of operations have essentially the same ``power'', in the sense just described. It has been shown in \cite{Gour17} that when restricted to states block-diagonal in energy that TP and TO coincide exactly in terms of state interconversion (both are governed by thermo-majorization), however at present it is unclear how/whether the interconversion structure of these two classes differ for states with coherence, and there is no obvious physical principle to choose one over the other. However, one nice aspect of TPs is that they admit a complete description in terms a single family of entropies that have a natural interpretation, while at present no complete set exists for TOs.

Note that the third assumption only accounts for \emph{external} coherences within the framework. The way in which we account for internal coherences in our analysis is to demand that the diagonal components (in the basis $|\mathbf{E} \rangle$) of the state are left invariant by the evolutions considered. Also note that if assumptions (1) and (3) hold then the isometry $\hat{W}$ can be taken to be equal to $\hat{V}$.

\section{Work extraction from a pure state coherence}
\subsection{Proof of Observation~1}
We find the conditions under which work can be deterministically extracted from coherence in a pure state.
After energy block diagonalizing, the state can be written as ${\cal D}(\ket{\Psi}\bra{\Psi}) = \sum_{\cal E} p_{\cal E} \ket{\psi_{\cal E}} \bra{\psi_{\cal E}}$, where $\ket{\psi_{\cal E}}$ are pure eigenstates of energy ${\cal E}$.
Suppose $\cal E^*$ gives the maximum value of $\log p_{\cal E} e^{\beta \cal E}$.
Then we have 
\begin{equation}
\begin{aligned}
W_{\rm coh} &= \inf_\alpha \left[ F_\alpha({\cal D}(\hat\rho) ) -  F_\alpha(\Pi(\hat\rho)) \right] \\
&= \inf_\alpha \left( \frac{1}{\alpha-1}  \right) 
\log \left[ \frac{\sum_{\cal E} p_{{\cal E}}^\alpha e^{-\beta(1-\alpha)\cal E} }
{\sum_{\cal E} e^{(1-\alpha) S_\alpha ( \hat\rho_{{\cal E}-\rm diag} ) } p_{{\cal E}}^\alpha e^{-\beta(1-\alpha)\cal E} } \right],
\end{aligned}
\end{equation}
where $\hat\rho_{{\cal E}-\rm diag} = \Pi( \ket{\psi_{\cal E}} \bra{\psi_{\cal E}})$ is a fully dephased state in the energy eigenspace ${\cal E }$.  
Then we notice that $S_\alpha ( \hat\rho_{{\cal E}-\rm diag} ) )  > 0$ for any $\alpha \in [0, \infty)$, unless $\hat\rho_{{\cal E}-\rm diag}$ is incoherent (i.e. no internal coherence for $\cal E$).
This leads to 
$ F_\alpha({\cal D}(\hat\rho) ) -  F_\alpha(\Pi(\hat\rho)) >0 $ for any finite value of $\alpha$.
In the limit $\alpha \rightarrow \infty$, $F_\infty ({\cal D}(\hat\rho) ) -  F_\infty (\Pi(\hat\rho)) = \log p_{\cal E^*} e^{\beta \cal E^*} - \max_{{\cal E}, \lambda_{\cal E} } p_{\cal E} \lambda_{\cal E} e^{\beta \cal E} > 0$, unless $\hat\rho_{{\cal E^*}-\rm diag}$ is incoherent. Here, $\lambda_{\cal E}$ are eigenvalues of $\hat\rho_{{\cal E}-\rm diag}$.
 Thus if a pure state does not contain internal coherence for ${\cal E}^*$, $W_{\rm coh} \leq F_\infty ({\cal D}(\hat\rho) ) -  F_\infty (\Pi(\hat\rho)) =0$.
 Conversely, if the state contains internal coherence for ${\cal E}^*$, 
$F_\alpha ({\cal D}(\hat\rho) ) -  F_\alpha (\Pi(\hat\rho)) >0$ for all
$\alpha \in [0,\infty)$, thus positive work can be extracted.

\subsection{Work extraction condition for a bipartite two-level system}
We consider work extraction from coherence in a two-qubit system with local energy difference $\omega_0$ in each qubit, starting from a pure state of the form
$$
\ket{\psi} = \sqrt{p_0} \ket{00} +  \sqrt{p_1} \left( \frac{ \ket{01} + \ket{10}}{ \sqrt{2} } \right)+ \sqrt{p_2} \ket{11} .
$$
Observation 1 says that $p_1$ should be large enough to extract work under a single-shot thermal operation.
In this case, the condition from Observation 1 can be written  as
$$
\begin{cases}
p_1 e^{\beta \omega_0} > p_0 = 1 - p_1 - p_2 \\
p_1 e^{\beta \omega_0} > p_2 e^{2\beta\omega_0} .
\end{cases}
$$
This leads to a necessary condition for extracting a positive amount of work $W_{\rm ext} > 0$ from coherence:
$$
p_1 > \frac{1}{1+e^{\beta\omega_0} + e^{-\beta\omega_0}}
$$

Thus if $0< p_1 < (1+e^{\beta\omega_0} + e^{-\beta\omega_0})^{-1}$, we cannot extract work from coherenc,e even though the state definitely contains internal coherence of the form $\ket{01} + \ket{10}$.
On the other hand, a sufficient condition for $W_{\rm ext} >0$ can be obtained:
$$
p_1 > \frac{e^{\beta\omega_0}}{1+e^{\beta\omega_0}}.
$$

\section{An example of the quantum Fisher (skew) information imposing independent constraints from $F_{\alpha}$ or $A_\alpha$}
We present an example showing that our asymmetry quantifiers give constraints on quantum thermodynamics independent from those due to the free energies $F_\alpha$ or the coherence measures $A_\alpha$. Let us consider the transformation by a thermal process of the initial state
$$\hat\rho =
\left(
\begin{array}{cccc}
 0.5 & 0 & 0.1 & 0.1 \\
 0 & 0.2 & 0 & 0 \\
 0.1 & 0 & 0.25 & 0.1 \\
 0.1 & 0 & 0.1 & 0.05 \\
\end{array}
\right)
$$
to the final state
$$
\hat\sigma = 
\left(
\begin{array}{cccc}
 0.5 & 0.099 & 0.099 & 0.099 \\
 0.099 & 0.25 & 0 & 0 \\
 0.099 & 0 & 0.2 & 0 \\
 0.099 & 0 & 0 & 0.05 \\
\end{array}
\right)
$$
with the Hamiltonian $\hat{H} = \sum_{n=0}^3 n \omega \ket{n}\bra{n}$. It can be checked that the free energies $F_\alpha$ and coherence measures $A_\alpha$ of the initial state $\hat\rho$ are larger than those of the final state $\hat\sigma$. Furthermore, each mode of coherence is decreased from $0.1$ to $0.099$. However, the skew information values for $\alpha=1/2$ are given by $I_{1/2}(\hat\rho,\hat{H}) = 0.153$ and $I_{1/2}(\hat\sigma,\hat{H}) = 0.163$; the quantum Fisher information values are $I_F(\hat\rho,\hat{H}) = 0.843$  and $I_F(\hat\sigma,\hat{H}) = 0.959$ (all in units of $\omega^2$). Thus a thermal process cannot transform $\hat\rho$ into $\hat\sigma$, but such a transformation is not disallowed by the restrictions given by $F_\alpha$ or $A_\alpha$.

This is due to that the coherence monotones given by the quantum Fisher information and skew information capture not only the degree of coherence between different energy eigenstates, but also take account of  how much energy level spacing exists in each coherence term.

\section{Clock/work trade-off relation: Two-level local Hamiltonian systems}
 We first prove the following proposition.
 \begin{proposition}[Work bound] For a given energy distribution $p_{\cal E}$, the extractable work from coherence is upper bounded as follows:
\begin{equation}
W_{\rm coh} \leq k_B T \sum_{\cal E} p_{\cal E} \log g_{\cal E},
\end{equation}
where $g_{\cal E}$ is the dimension of the eigenspace of energy level ${\cal E}$.
\label{Prop1}
\end{proposition}

\begin{proof}
 Note that
\begin{equation}
\begin{aligned}
W_{\rm coh} &= \inf_\alpha \left[ F_\alpha({\cal D}(\hat\rho)) - F_\alpha( \Pi(\hat\rho)) \right] \\
&\leq F({\cal D}(\hat\rho)) - F( \Pi(\hat\rho)) \\
&= k_B T \left[ S( \Pi(\hat\rho))  - S({\cal D}(\hat\rho)) \right].
\end{aligned}
\end{equation}
Since both $\Pi(\hat\rho)$ and ${\cal D}(\hat\rho)$ are energy-block diagonal, we can express 
$\Pi(\hat\rho) = \sum_{{\cal E}, \lambda} p^\Pi_{{\cal E}, \lambda} \ket{{\cal E}, \lambda}\bra{{\cal E}, \lambda}$ and ${\cal D}(\hat\rho)= \sum_{{\cal E}, \lambda} p^{\cal D}_{{\cal E}, \lambda} \ket{{\cal E}, \lambda}\bra{{\cal E}, \lambda}$ for $\lambda = 1, 2, \cdots, g_{\cal E}$ with $\sum_{\lambda=1}^{g_{\cal E}} p^\Pi_{{\cal E}, \lambda} = \sum_\lambda p^{\cal D}_{{\cal E}, \lambda} =p_{\cal E}$.  Then we have
\begin{equation}
\begin{aligned}
S(\Pi(\hat\rho)) - S({\cal D}(\hat\rho)) &= \sum_{\cal E} p_{\cal E} \sum_{\lambda=1}^{g_{\cal E}} \left[ \frac{p^{\cal D}_{{\cal E}, \lambda}}{p_{\cal E}} \log \frac{p^{\cal D}_{{\cal E}, \lambda}}{p_{\cal E}} - \frac{p^{\Pi}_{{\cal E}, \lambda}}{p_{\cal E}} \log \frac{p^{\Pi}_{{\cal E}, \lambda}}{p_{\cal E}} \right]  \\
&\leq \sum_{\cal E} p_{\cal E} \log g_{\cal E},
\end{aligned}
\end{equation}
since $S(\hat\rho) - S(\hat\sigma) \leq \log d$ for $d$-dimensional states $\hat\rho$ and $\hat\sigma$.
\end{proof}

\hfill
\subsection{Proof of Theorem~1}
Here we provide a complete proof of Theorem~1 that extractable work from coherence is upper bounded by the quantum Fisher information:
\begin{equation}
W_{\rm coh} \leq k_B T N (\log 2) H_b\left( \frac{1}{2} \left[1-\sqrt{\frac{I_F(\hat\rho, \hat{H})}{N^2 \omega_0^2}} \right] \right).
\end{equation}
\begin{proof}
In an $N$-particle two-level system with energy difference $\omega_0$, the degeneracy of the energy level ${\cal E}$ is given by 
$g_{\cal E} = \binom{N}{n}$, where ${\cal E} = \omega_0 n$.
By using the fact that
$$ \binom{N}{n} \leq 2^{ N H_b \left(\frac{n}{N} \right)}$$
for every $N$ and $n$, we obtain
\begin{equation}
W_{\rm coh} \leq k_B T \sum_{\cal E} p_{\cal E} \log \binom{N}{n} \leq N k_B T ( \log 2 ) \sum_{\cal E} p_{\cal E}  H_b(n/N),
\label{ExtBound1}
\end{equation}
where Proposition~1 has been applied to obtain the first inequality.

Furthermore, we can express the binary entropy as
$$
H_b(x) = 1 - \frac{1}{2 \log 2} \sum_{j=1}^\infty \frac{(1-2x)^{2j}}{j(2j-1)}.
$$
For a given probability distribution $p_x$ and $j \geq 1$ we have
$$
 \sum_x p_x (1-2x)^{2j} \geq \left[\sum_x p_x (1- 2x)^2 \right]^j = (1- 2y)^{2j},
$$
where $y = \frac{1}{2} \left[ 1 \pm \sqrt{(1-2\bar{x})^2 + 4 {\rm Var}_x} \right]$
with $\bar{x} = \sum_x p_x x$ and ${\rm Var}_x = \sum_x p_x (x - \bar{x})^2$.
Then we have 
$$
\begin{aligned}
\sum_x p_x H_b(x) &= 1 - \frac{1}{2 \log 2} \sum_{j=1}^\infty \sum_x p_x \frac{(1-2x)^{2j}}{j(2j-1)} \\
& \leq 1 - \frac{1}{2 \log 2} \sum_{j=1}^\infty \frac{(1-2y)^{2j}}{j(2j-1)} \\
& = H_b(y).
\end{aligned}
$$
By substituting this result into Eq.~(\ref{ExtBound1}), we obtain

\begin{equation}
\label{BoundN}
\begin{aligned}
W_{\rm coh} & \leq N k_B T (\log 2) H_b \left( \frac{1}{2} \left[ 1 \pm \sqrt{ \left(1- \frac{2 \bar{E}}{N\omega_0}\right)^2 + \frac{ {4 \rm Var}_{\hat{H}}}{N^2 \omega_0^2} } \right]\right),
\end{aligned}
\end{equation}
where $\bar{E} = \langle \hat{H} \rangle_{\hat\rho}$ and ${\rm Var}_{\hat{H}} = \langle (\hat{H} - \bar{E})^2 \rangle_{\hat\rho} $.
Note that $H_b$ is symmetric about $x=1/2$ and monotonically increasing for $x \leq 1/2$.
We also note that $4 {\rm Var}_{\hat{H}} \geq I_F(\hat\rho, \hat{H})$ for every quantum state $\hat{\rho}$.
These observations lead to 
$H_b \left( \frac{1}{2} \left[ 1 \pm \sqrt{ \left(1- \frac{2 \bar{E}}{N\omega_0}\right)^2 + \frac{ {4 \rm Var}_{\hat{H}}}{N^2 \omega_0^2} } \right] \right) \leq H_b \left( \frac{1}{2} \left[ 1 - \sqrt{  \frac{ I_F(\hat\rho, \hat{H})}{N^2 \omega_0^2} } \right] \right)$, which completes the proof.
\end{proof}

\subsection{Tighter bound of the trade-off relation}
The bound from Theorem~1 can be tightened.
We have observed that
\begin{equation}
\label{BinomBound}
\binom{N}{rN} \leq \binom{N}{N/2}^{H_b(r)}
\end{equation}
or equivalently
$$
\log \binom{N}{rN} \leq  {H_b(r)} \log \binom{N}{N/2}
$$
for every $0 \leq r \leq 1$ and $N$ up to $N=100$. The binomial coefficient for an odd number of $N$ is defined as $\binom{N}{N/2} := \frac{\Gamma(N+1)}{\Gamma (N/2+1)^2}$ by using the Gamma function $\Gamma(z) = \int_0^{\infty} x^{z-1} e^{-x} dx$.
Figure~\ref{BinomialFig} shows that the inequality (\ref{BinomBound}) well holds for $N \leq100$ and seemingly holds for every number of $N$, yet the proof for a general case has not been found.

\begin{figure}[t]
\includegraphics[width=0.43\linewidth]{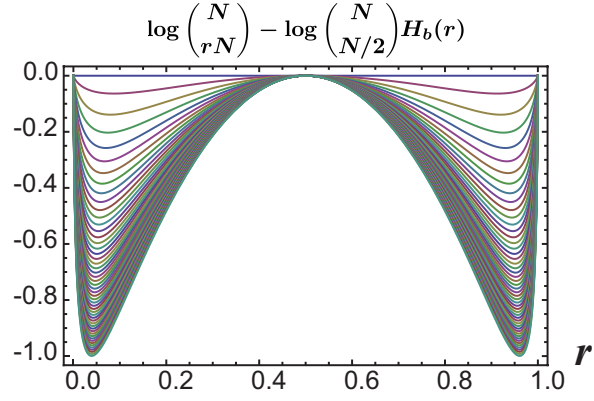}
\caption{Numerical verification of the bound $\log \binom{N}{n} \leq \log \binom{N}{N/2} H_b(n/N)$ for $N \leq 100$.
The value of $N$ increases from the upper most curve ($N=2$) to the lowest curve ($N=100$).
}
\label{BinomialFig}
\end{figure}

From the inequality above, we obtain a tighter bound of the trade-off relation \ref{BoundN} by taking $n=rN$ and replacing $N \log 2$ with $\log \binom{N}{N/2}$:
\begin{equation}
\begin{aligned}
W_{\rm coh} & \leq k_B T \log \binom{N}{N/2} H_b \left( \frac{1}{2} \left[ 1 - \sqrt{ \frac{ I_F(\hat\rho, \hat{H})}{N^2 \omega_0^2} } \right]\right),
\end{aligned}
\end{equation}
When $N \gg 1$, we note that $N \log 2 \approx \log \binom{N}{N/2}$ then the bound approaches to the bound of Theorem~1.

\subsection{Clock/work trade-off for $N=2$}
We show the tighter trade-off relation between clock/work resources for $N=2$ case:
\begin{equation}
W_{\rm coh} + ( k_B T  \log 2) \left( \frac{ I_F(\hat\rho, \hat{H})}{4\omega_0^2} \right) \leq k_B T  \log 2 .
\end{equation}

\begin{proof}
Suppose the state has probability $p_0$, $p_1$, and $p_2$ for each energy level $0$, $\omega_0$ and $2\omega_0$.
By using Eq.~(\ref{ExtBound1}) for $N=2$, we have $W_{\rm coh} \leq k_B T (\log 2) p_1$, since the state has a degeneracy in the energy-eigenspace only for ${\cal E} = \omega_0$.
In this case, energy variance ${\rm Var}_{\hat{H}}$ is given by
$$
{\rm Var}_{\hat{H}} = \omega_0^2 (-p_1^2 + p_1 - 4p_2^2 + 4p_2 - 4 p_1 p_2),
$$
which leads to the maximum value of $p_1$,
$$
p_1^{\rm max} = 1 - {\rm Var}_{\hat{H}} / \omega_0^2,
$$
for a given value of ${\rm Var}_{\hat{H}}$.
Again, we can use $4 {\rm Var}_{\hat{H}} \geq I_F(\hat\rho,\hat{H})$ to get
$$
W_{\rm coh} \leq k_B T (\log 2) p_1^{\rm max} =  k_B T (\log 2) \left( 1 - \frac{I_F(\hat\rho, \hat{H})}{4 \omega_0^2} \right).
$$
which is the desired inequality.
\end{proof}

\section{Clock/work trade-off relation: General case}
\subsection{Proof of Theorem~2}
We prove the statement of Theorem~2:
\begin{equation}
W_{\rm coh} + k_B T \left( \frac{I_F(\hat\rho,\hat{H})}{2 \Delta_E^2} \right) \leq k_B T \sum_{i=1}^N \log d^{(i)},
\end{equation}
where $\Delta_E^2 = \sum_{i=1}^N (\Delta_E^{(i)})^2 $ with $\Delta_E^{(i)}$ is the maximum energy difference of the $i$th subsystem.
\begin{proof}
In this case, the degeneracy $g_{\cal E}$ of the energy ${\cal E}$ is given by
$$
g_{\cal E} = \prod_{i=1}^N d^{(i)} f_{\cal E},
$$
where $f_{\cal E} = \sum_{{\cal E}_{\boldsymbol E} = {\cal E}} P(\boldsymbol E)$ is a probability (or frequency) to have the energy $\cal E$ in the $N$-particle system, since $\prod_{i=1}^N d^{(i)}$ is total possible numbers of $\boldsymbol E$.
Then $f_{\cal E}$ can be considered as a probability distribution of a variable
$X_N = \sum_{i=1}^N E_i$
from the distribution of independent random variables of $E_i$ for $i$th party.
In our case, $E_i$ is strictly bounded by $E^{(i)}_1  \leq E_n \leq E^{(i)}_{d^{(i)}}$ and it has the same probability $P(E_i = E^{(i)}_j) = 1/d^{(i)}$ for every $j=1,2,\cdots, d^{(i)}$ and zero for all other cases.
Hoeffding's inequality \cite{Heffding63}, then shows that
$$
P(X_N- \mu_E \geq t ) \leq \exp\left[{-\frac{2 t^2} {\Delta_E^2}}\right],
$$
where $\displaystyle \mu_E :=  \mathbb{E} (X_N) = \mathbb{E} \left( \sum_{i=1}^N E_i \right)   = \sum_{i=1}^N \sum_{j=1}^{d^{(i)}} \left( \frac{E^{(i)}_j}{d^{(i)}} \right)$ and  $\Delta_E^2 = \sum_{i=1}^N (\Delta_E^{(i)})^2 $.
Using this, the upper bound of $f_{\cal E}$ is given by
$$
f_{\cal E} = P(  X_N = {\cal E}) \leq P(  X_N \geq {\cal E}) \leq \exp \left[ -\frac{2( {\cal E} - \mu_E )^2 }{\Delta_E^2} \right].
$$

Again using Proposition 1, we have
\begin{equation}
\begin{aligned}
W_{\rm coh} & \leq k_B T \sum_{\cal E} p_{\cal E} \log g_{\cal E} \\
& = k_B T \sum_{\cal E} p_{\cal E} \log \left( \prod_{i=1}^N d^{(i)} f_{\cal E} \right) \\
&= k_B T \sum_{i=1}^N \log d^{(i)} + k_B T \sum_{\cal E} p_{\cal E} \log f_{\cal E} \\
&\leq k_B T \sum_{i=1}^N \log d^{(i)} - \frac{2 k_B T}{\Delta_E^2} \sum_{\cal E} p_{\cal E} ({\cal E} - \mu_E)^2 \\
&= k_B T \sum_{i=1}^N \log d^{(i)} - \frac{2 k_B T}{\Delta_E^2} {\rm Var}_{\hat{H}} - \frac{2 k_B T}{\Delta_E^2} (\bar{E} - \mu_E)^2  \\
&\leq k_B T \sum_{i=1}^N \log d^{(i)} - k_B T \left( \frac{I_F(\hat\rho, \hat{H})}{ 2 \Delta_E^2} \right),
\end{aligned}
\end{equation}
where the last inequality is from the fact $4 {\rm Var}_{\hat{H}} \geq I_F(\hat\rho, \hat{H})$.
\end{proof}

\subsection{Trade-off relation by allowing a small energy resolution window}
In a real experimental situation, it is hard to access an exact energy level with prefect prevision, so we may permit a finite energy gap $\epsilon$ in energy levels.
Under this assumption, we view states with an $\epsilon$-energy gap to be ``essentially the same'' energy and so can carry internal coherences for approximate work extraction.
In order to introduce the energy gap $\epsilon$, we divide the energy spectrum into intervals
$$
{\cal E}_m = 
\begin{cases}
\left[\mu_E + \left(m-\frac{1}{2} \right) \epsilon, \mu_E + \left(m+\frac{1}{2} \right) \epsilon \right) &{\rm for~}m>0\\
\left(\mu_E + \left(m-\frac{1}{2} \right) \epsilon, \mu_E + \left(m+\frac{1}{2} \right) \epsilon \right] &{\rm for~}m<0\\
\left(\mu_E + \left(m-\frac{1}{2} \right) \epsilon, \mu_E + \left(m+\frac{1}{2} \right) \epsilon \right) &{\rm for~}m=0
\end{cases},
$$
where each interval has an energy width $\epsilon$.

Consequently, we define the energy distribution for $m$th interval $p_m^\epsilon = \sum_{{\cal E} \in {\cal E}_m} p_{\cal E}$ and its degeneracy $g_m^\epsilon =  \sum_{{\cal E} \in {\cal E}_m} g_{\cal E}$, respectively.
If we allow the $\epsilon$ energy gap for internal coherences, the amount of extractable work is upper bounded by
\begin{equation}
\begin{aligned}
W_{\rm coh}^\epsilon &\leq k_B T \sum_m p^\epsilon_m \log g^\epsilon_m \\
&=  k_B T \sum_{i=1}^N \log d^{(i)} +  k_B T\sum_m p^\epsilon_m \log f^\epsilon_m,
\end{aligned}
\end{equation}
where $f_m^\epsilon$ is the frequency to be in the $m$th energy interval. 
The upper bound of $f_m^\epsilon$ is then given by
$$
f_m^\epsilon = P ( X_N \in {\cal E}_m) \leq P\left(|X_N -  (\mu_E +   \epsilon m)| \geq \frac{1}{2}\epsilon \right) \leq 
\begin{cases}
\exp\left[{-\frac{2(|m|-1/2)^2 \epsilon^2}{\Delta^2_E}}\right] & (m \neq 0)\\
1 & (m=0)
\end{cases}.
$$
By following the same argument with Theorem~2, we have 
\begin{equation}
\begin{aligned}
\label{EpsilonI}
W_{\rm coh}^\epsilon &\leq  k_B T \sum_{i=1}^N \log d^{(i)} -  k_B T \sum_{m\neq0} p^\epsilon_m \left[ \frac{2(|m|-\frac{1}{2})^2\epsilon^2}{\Delta_E^2} \right] \\
 &=  k_B T \sum_{i=1}^N \log d^{(i)}  -   \frac{2k_B T}{\Delta_E^2} \left[ \sum_m p^\epsilon_m |m|^2 \epsilon^2 -  \epsilon \sum_{m} p^\epsilon |m| +\frac{1}{4} \epsilon^2 \sum_{m \neq 0}p_m^\epsilon \right].
\end{aligned}
\end{equation}

When ${\cal E} \in {\cal E}_m$, we use the fact that  $\mu_E + (m - 1/2) \epsilon \leq {\cal E} \leq \mu_E + (m + 1/2) \epsilon$ to show 
$$
\begin{aligned}
\sum_{\cal E} p_{\cal E} ({\cal E} - \mu_E)^2 &= \sum_m \sum_{{\cal E} \in {\cal E}_m} p_{\cal E} ({\cal E} - \mu_E)^2 \\
&\leq \sum_m \left(  \sum_{{\cal E} \in {\cal E}_m} p_m^\epsilon \right) \left( m^2\epsilon^2 + |m| \epsilon + \frac{1}{4} \epsilon^2 \right) \\
&= \sum_m p_m^\epsilon m^2 \epsilon^2 + \epsilon \sum_m p^\epsilon_m |m| + \frac{1}{4} \epsilon^2.
\end{aligned}
$$
By substituting this inequality into (\ref{EpsilonI}), then we finally get the following trade-off relation by allowing a small energy gap $\epsilon$, 
\begin{equation}
W_{\rm coh}^\epsilon + k_B T \left( \frac{I_F(\hat\rho,\hat{H})}{2 \Delta_E^2} \right) \leq k_B T \sum_{i=1}^N \log d^{(i)} + R(\epsilon),
\end{equation}
where the correction term is given by
$$
R(\epsilon) = \frac{k_B T}{\Delta_E^2} \left[ {2 \epsilon} \sum_m p^\epsilon_m |m| - \epsilon^2 p_0^\epsilon \right].
$$

We additionally analyse how the quantum Fisher information is perturbed under the same energy resolution window. 
In this case, the Hamiltonian corresponds to the energy levels $\{ {\cal E}_m\}$ can be expressed as
$\hat{H} + (\epsilon/2) \hat{H}_I$, where $(\epsilon/2)\hat{H}_I$  with $|| \hat{H}_I ||_\infty \leq 1$ fills the $\epsilon$-energy gaps of $\hat{H}$.
Then the quantum Fisher information with respect to this Hamiltonian is given by
\begin{equation}
\label{QFIperturb}
\begin{aligned}
I_F^\epsilon (\hat\rho, \hat{H}) &= I_F\left(\hat\rho, \hat{H} + \frac{\epsilon}{2} \hat{H}_I\right)\\
&=2 \sum_{j,k} \frac{(\lambda_j - \lambda_k)^2}{\lambda_j + \lambda_k} |\bra{\psi_j} \hat{H} + \frac{\epsilon}{2} \hat{H}_I \ket{\psi_k}|^2\\
&= 2 \sum_{j,k} \frac{(\lambda_j - \lambda_k)^2}{\lambda_j + \lambda_k}  \left[ |\bra{\psi_j} \hat{H}\ket{\psi_k}|^2 + \epsilon \bra{\psi_j} \hat{H}\ket{\psi_k}\bra{\psi_k} \hat{H}_I\ket{\psi_j}+ \frac{\epsilon^2}{4}  |\bra{\psi_j} \hat{H}_I \ket{\psi_k}|^2 \right]\\
&= I_F(\hat\rho, \hat{H}) + 2\epsilon \sum_{j,k} \frac{(\lambda_j - \lambda_k)^2}{\lambda_j + \lambda_k}\bra{\psi_j} \hat{H}\ket{\psi_k}\bra{\psi_k} \hat{H}_I\ket{\psi_j} + \frac{\epsilon^2}{4} I_F(\hat\rho, \hat{H}_I),
\end{aligned}
\end{equation}
where  $\lambda_j$ and $\ket{\psi_j}$ are eigenvalue and eigenstate of $\hat\rho$, respectively.

By using the fact $\left|\sum_{j,k} \frac{(\lambda_j - \lambda_k)^2}{\lambda_j + \lambda_k}\bra{\psi_j} \hat{H}\ket{\psi_k}\bra{\psi_k} \hat{H}_I\ket{\psi_j} \right| \leq \left(\sum_{j,k} \frac{(\lambda_j - \lambda_k)^2}{\lambda_j + \lambda_k} \right) ||\hat{H}||_\infty ||\hat{H}_I||_\infty \leq 2 ||\hat{H}||_\infty $, we observe that the $\epsilon$-energy resolution window perturbs the quantum Fisher information at most $4 \epsilon ||\hat{H}||_\infty + \epsilon^2$. Note that $I_F(\hat\rho, \hat{H}_I) \leq 4 ||\hat{H}_I||_\infty^2 = 4$.
Finally, we obtain the following trade-off relation between clock/work resources with the $\epsilon$-energy resolution
\begin{equation}
W_{\rm coh}^\epsilon + k_B T \left( \frac{I_F^\epsilon(\hat\rho,\hat{H})}{2 \Delta_E^2} \right) \leq k_B T \sum_{i=1}^N \log d^{(i)} + \tilde{R}(\epsilon),
\end{equation}
where $\tilde{R}(\epsilon) = \frac{k_B T}{\Delta_E^2} \left[ 2 \epsilon \left(\sum_m p^\epsilon_m |m| + ||\hat{H}||_\infty \right) + \epsilon^2 (\frac{1}{2}-p_0^\epsilon) \right] = {\cal O}(\epsilon)$.

\section{Nonlocal energy-diagonal states problem and energy level degeneracy in the Ising model}

The problem of nonlocal energy-diagonal states can be also circumvented by allowing a small energy gap when interaction is weak.
As an example, we analyse this in the Ising model.
In the 1D transverse-field Ising model, the Hamiltonian is given by 
$$\hat{H}_{\rm Ising} = - h \sum_{i=1}^N \hat{\sigma}_z^{(i)} - J \sum_{i=1}^N \hat{\sigma}_x^{(i)} \otimes \hat{\sigma}_x^{(i+1)},$$
 where $h$ and $J$ are the strength of the transverse field and the coupling between adjacent spins, respectively. Here $\hat{\sigma}^{(i)}_{x,y,z}$ are the Pauli-$x,y,z$ operators for $i$th spin, and we additionally assume a periodic boundary condition $\hat{\sigma}_{x,y,z}^{(N+1)} = \hat{\sigma}_{x,y,z}^{(1)}$. 
We use $\ket{0}$ and $\ket{1}$,  which are eigenstates of $\hat\sigma_z$ with eigenvalue $-1$ and $+1$, respectively, as a computational basis of local spins.

We first consider the case of $N=2$, with a weak interaction $J \ll h$. 
When there is no interaction, i.e. $J=0$, an energy eigenvalue $E=0$ has degenerate eigenstates $\ket{01}$ and $\ket{10}$.
For $J \neq 0$, the energy level splits into $E = \pm 2J$, where corresponding eigenstates are given by the entangled states $\ket{ \psi_\pm }= \ket{01} \mp \ket{10}$.
Assuming a weak interaction $J \ll h$, we can allow a small energy gap $\epsilon > 4J$, then we can consider $\ket{ \psi_\pm }$ to have effectively the same energy level and transformation between them is possible by a thermal process.
In this case, {\it internal coherence} is now contained in the coherence between $\ket{ \psi_\pm }$. 
A thermal process can transform this type of internal coherences to a fully mixed separable form $\ket{01}\bra{01} + \ket{10}\bra{10}$, by which we can extract extra work as in the noninteracting Hamiltonian.

The same technique of (\ref{QFIperturb}) can be used to calculate perturbation of the QFI by the weak interaction,
$$
\begin{aligned}
I_F(\hat\rho, \hat{H}_{\rm lsing}) &= 2 \sum_{j,k} \frac{(\lambda_j - \lambda_k)^2}{\lambda_j + \lambda_k} |\bra{\psi_j} \hat{H}_{\rm Ising} \ket{\psi_k}|^2\\
&= 2 \sum_{j,k} \frac{(\lambda_j - \lambda_k)^2}{\lambda_j + \lambda_k} \\
&\qquad \times \left[ |\bra{\psi_j} \hat{H}_0 \ket{\psi_k}|^2 +  2hJ \bra{\psi_j}\sum_{i=1}^N \hat{\sigma}_x^{(i)} \ket{\psi_k} \bra{\psi_k} \sum_{i=1}^N \hat{\sigma}_x^{(i)} \otimes \hat{\sigma}_x^{(i+1)} \ket{\psi_j}  + J^2 |\bra{\psi_j} \sum_{i=1}^N \hat{\sigma}_x^{(i)} \otimes \hat{\sigma}_x^{(i+1)} \ket{\psi_k}|^2\right]\\
&= I_F(\hat\rho, \hat{H}_0) + 4hJ\sum_{j,k} \frac{(\lambda_j - \lambda_k)^2}{\lambda_j + \lambda_k} \bra{\psi_j}\sum_{i=1}^N \hat{\sigma}_x^{(i)} \ket{\psi_k} \bra{\psi_k} \sum_{i=1}^N \hat{\sigma}_x^{(i)} \otimes \hat{\sigma}_x^{(i+1)} \ket{\psi_j} + J^2 I_F\left(\hat\rho, \sum_{i=1}^N \hat{\sigma}_x^{(i)} \otimes \hat{\sigma}_x^{(i+1)}\right)
\end{aligned}
$$
where $\hat{H}_0 = - h \sum_{i=1}^N \hat{\sigma}_z^{(i)}$ is the unperturbed Hamiltonian.
Then with the interaction, we have 
$$
I_F(\hat\rho, \hat{H}_0) - 8hJN^2 \leq I_F(\hat\rho, \hat{H}_{\rm lsing}) \leq I_F(\hat\rho, \hat{H}_0) + 8 h J N^2 + 4 J^2 N^2.
$$
by using $||\sum_{i=1}^N \hat{\sigma}_x^{(i)}||_\infty \leq N$ and $||\sum_{i=1}^N \hat{\sigma}_x^{(i)} \otimes \hat{\sigma}_x^{(i+1)}||_\infty \leq N$.
%We also point out that $I_F(\hat\rho, \hat{H}_0) \propto h^2$.
%Thus, both the work and clock resources have small perturbation in the presence of weak interaction, $J \ll h$ when assuming a $\epsilon$-energy resolution window with $4J < \epsilon \ll h$.
 
 For general $N$ and an arbitrary strength of interaction $J$, the Ising Hamiltonian can be rewritten in terms of non-interacting quasi-particles \cite{Sachdev11}. In this viewpoint, we can redefine internal and external coherences between the non-interacting quasi-particles and apply Theorem~2 accordingly.

In the transverse Ising model, the spin operators can be mapped to the excitations of spinless fermions by the Jordan-Wigner transformation
$\hat{c} = \left(\bigotimes_{j=1}^{l-1} \hat{\sigma}_z^{(j)}\right)  \hat{\sigma}_-^{(l)}$, where $\hat{\sigma}_\pm^{(l)} = (\hat{\sigma}_x^{(l)} \pm i \hat{\sigma}_y^{(l)})/2$. Note that the operator $\hat{c}_k$ and $\hat{c}^\dagger_l$ satisfies the following anti-commutation relation: $\{ \hat{c}_k, \hat{c}^\dagger_l \} = \delta_{k,l}$ and $\{ \hat{c}_k, \hat{c}_l \} = 0$.
Then we note that the Ising Hamiltonian can be divided into two block diagonal Hamiltonians
$\hat{H}_{\rm Ising} = {\hat H}_e \oplus {\hat H}_o$,
depending on the parity of the fermion number $\hat{N} = \sum_j \hat{c}_j^\dagger \hat{c}_j$ \cite{IsingRef}.

Each Hilbert space can be diagonalized by taking Fourier transform of fermion operators,
$$
\hat{c}_{k_n} = \frac{1}{\sqrt{N}} \sum_{j=1}^N \hat{c}_j e^{ i k_n j}~{\rm (even)} \quad {\rm and} \quad \hat{c}_{p_n} = \frac{1}{\sqrt{N}} \sum_{j=1}^N \hat{c}_j e^{ i p_n j}~{\rm (odd)},
$$
where
$$
k_n = \frac{2\pi( n + 1/2)}{N}~{\rm (even)} \quad {\rm and} \quad p_n = \frac{2 \pi n} {N}~{\rm (odd)}
$$
for $n = -\frac{N}{2}, \cdots , \frac{N}{2}-1$.
Then the following Bogoliubov transformation for both $k_n$ and $p_n$
$$
\begin{aligned}
\hat{c}_k &= \cos (\theta_k /2) \hat{b}_k + i \sin(\theta_k /2) \hat{b}_{-k}^\dagger \\
\hat{c}_k^\dagger &= i \sin (\theta_k /2) \hat{b}_k + \cos (\theta_k /2) \hat{b}_{-k}^\dagger,
\end{aligned}
$$
with the Bogoliubov angle $e^{i \theta_k} = (h-J e^{ik})/\sqrt{J^2 + h^2 -2 h J \cos k}$
leads to 
$$
\hat{H}_{e} = \sum_{n} {\cal E}_{k_n} \left[ \hat{b}_{k_n} ^\dagger \hat{b}_{k_n} - \frac{1}{2} \right]~{\rm (even)} \quad {\rm and} \quad \hat{H}_{o} = \sum_{n\neq0} {\cal E}_{p_n} \left[ \hat{b}_{p_n} ^\dagger \hat{b}_{p_n} - \frac{1}{2} \right] - 2(J-h) \left[ \hat{b}^\dagger_{0}\hat{b}_{0} - \frac{1}{2} \right]~{\rm (odd)},
$$
with the dispersion relation \cite{Sachdev11, IsingRef}
$$
{\cal E}_k = 2 \sqrt{ h^2 + J^2 - 2 h J \cos (k)}.
$$
Thus the energy spectrum of the Hamiltonian $\hat{H}_{\rm Ising}$ can be evaluated as
${\cal E}(k_1, \cdots, k_{2m})$ for even numbers of the fermion excitation 
$\ket{k_1, \cdots, k_{2m}} = \prod_{i=1}^{2m} \hat{b}^\dagger_{k_i} \ket{0}_{NS} \in {\cal H}_e$
known as Neveu-Schwarz (NS) sector
and 
${\cal E}(p_1, \cdots, p_{2m})$ for odd numbers of the fermion excitation
$\ket{p_1, \cdots, p_{2m+1}} = \prod_{i=1}^{2m+1} \hat{b}^\dagger_{p_i} \ket{0}_{R} \in {\cal H}_o$
known as Ramond (R) sector.
Figure~\ref{IsingFig} shows the number of degeneracies in the energy eigenspaces with various configurations of $J$ and $h$ for $N=16$.

In this case of large $N$, we note that $k_n \approx p_n$ and the diagonalized form of the Ising Hamiltonian 
can be considered as the $N$-particle non-interacting two-level Hamiltonian model with energy levels $\left\{ -\frac{{\cal E}_{k_n}}{2}, \frac{{\cal E}_{k_n}}{2} \right\}$
in the $n$th site.
Then we can apply Theorem~2 to obtain the clock/work trade-off relation
by taking
$$
\Delta_E^2 = \sum_{n=1}^N (\Delta_E^{(n)})^2 \approx \sum_{n=1}^N {\cal E}_{k_n}^2 = \sum_{n=1}^N 4 ( h^2 + J^2 - 2 h J \cos (k_n)) = 4N ( h^2 + J^2).
$$
and $d^{(n)} = 2$ for every $n$.
Finally, we have the trade-off relation for the 1D-transverse Ising model:
$$
W_{\rm coh} + k_B T \left( \frac{I_F(\hat\rho,\hat{H})}{8N (h^2 + J^2)} \right) \leq N k_B T \log 2.
$$
For general nonintegrable many-body systems, however, a general trade-off relation between clock/work resources is a highly non-trivial question, and deserves further study.
\begin{figure}[t]
\includegraphics[width=0.78\linewidth]{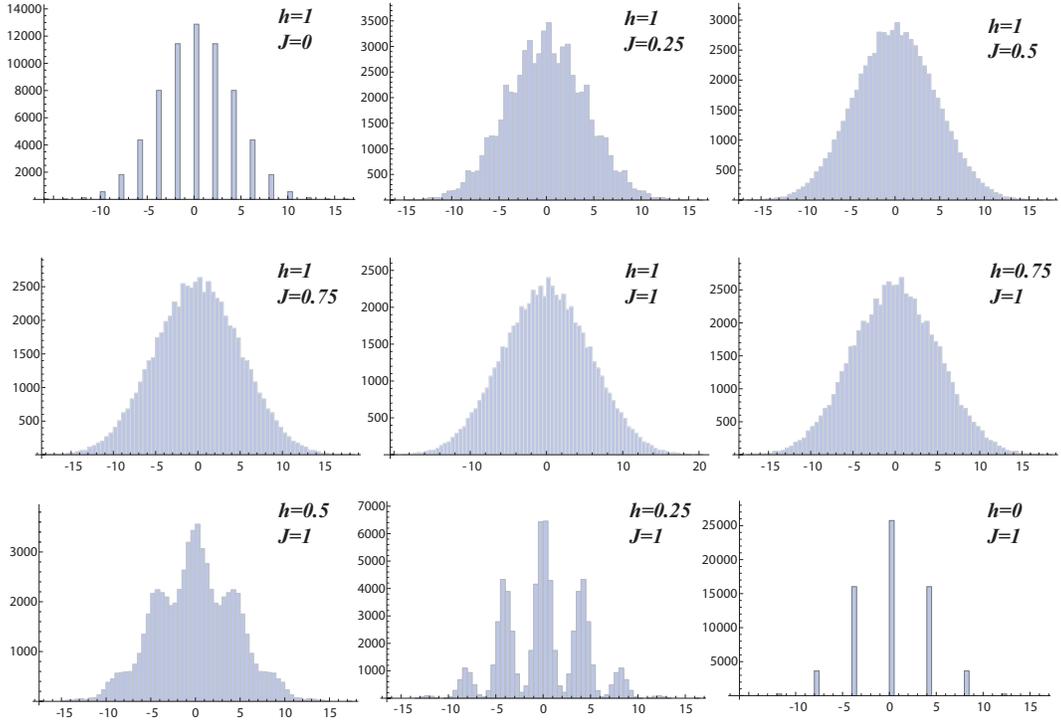}
\caption{Degeneracy in energy eigenspaces for the transverse Ising model. Parameters are chosen from $h=1, J=0$ to $h=0, J=1$ with $N=16$ by allowing the energy gap $\epsilon = 0.5$.}
\label{IsingFig}
\end{figure}

%\begin{thebibliography}{99}
%\bibitem{Gour17} G. Gour, D. Jennings, F. Buscemi, R. Duan, and I. Marvian, arXiv:1708.04302.
%\bibitem{Janzing2000} D. Janzing, P. Wocjan, R. Zeier, R. Geiss, T. Beth, Int. Journ. Th. Phys. \textbf{39} (12):2217-2753, (2000).
%\bibitem{Horodecki13} M. Horodecki and J. Oppenheim, Nat. Commun. {\bf 4}, 2059 (2013).
%\bibitem{Heffding63} W. Hoeffding, J. Am. Stat. Assoc. {\bf 58}, 13 (1963).
%\bibitem{Sachdev11} S. Sachdev, {\it Quantum Phase Transitions} (Cambridge University Press, Cambridge, 2011).
%\bibitem{IsingRef} E. Lieb, T. Schultz, and D. Mattis, Ann. Phys. (N.Y.) {\bf 16}, 407 (1961); P. Pfeuty,  Ann. Phys. (N.Y.) {\bf 57}, 79 (1970); P. Calabrese, F. H. L. Essler, M. Fagotti, J. Stat. Mech. (2012) P07016.
%\end{thebibliography}


\begin{thebibliography}{99}
\bibitem{footnote} Note that a ``work qubit" is a convenient unit to quantify deterministically extracted, perfectly ordered energy. It does not place restrictions on assumed work-bearing degrees of freedom and has been shown to be equivalent to other notions of work as ordered energy \cite{Faist12}.
\end{thebibliography}
\end{document}